\documentclass[conference,10pt]{IEEEtran}
\IEEEoverridecommandlockouts
\usepackage{amsmath,amsfonts,amsthm,amssymb,bbm}
\usepackage{cite}
\usepackage{balance}
\usepackage{mathrsfs}

\theoremstyle{plain} 
\newtheorem{theorem}{Theorem}
\newtheorem{corollary}{Corollary}
\newtheorem{definition}{Definition}
\newtheorem{lemma}{Lemma}

\theoremstyle{definition} \newtheorem{remark}{Remark}
\theoremstyle{definition} 

\title{Evaluating Multiple Guesses by an Adversary via a Tunable Loss Function}
\author{Gowtham R. Kurri, Oliver Kosut, Lalitha Sankar
\thanks{The authors are with the School of Electrical, Computer and Energy Engineering at Arizona State University. Email: {\tt gkurri@asu.edu, okosut@asu.edu, lsankar@asu.edu}}
\thanks{This work is supported in part by NSF grants CIF-1901243, CIF-1815361, and CIF-2007688.}
}
\begin{document}
\maketitle
\begin{abstract}
    We consider a problem of guessing, wherein an adversary is interested in knowing the value of the realization of a discrete random variable $X$ on observing another correlated random variable $Y$. The adversary can make multiple (say, $k$) guesses. The adversary's guessing strategy is assumed to minimize $\alpha$-loss, a class of tunable loss functions parameterized by $\alpha$. It has been shown before that this loss function captures well known loss functions including the exponential loss ($\alpha=1/2$), the log-loss ($\alpha=1$) and the $0$-$1$ loss ($\alpha=\infty$). We completely characterize the optimal adversarial strategy and the resulting expected $\alpha$-loss, thereby recovering known results for $\alpha=\infty$. We define an information leakage measure from the $k$-guesses setup and derive a condition under which the leakage is unchanged from a single guess.
    %
\end{abstract}
\section{introduction}
The classical guessing problem involves an adversary interested in finding the value of a realization of a discrete random variable $X$ by asking a series of questions in an adaptive manner until an affirmative answer is received. A commonly used performance metric for the guessing problem is the expected number of guesses required until $X$ is guessed correctly, or more generally a moment of this number. Massey~\cite{Massey} established a lower bound on the expected number of guesses in terms of the entropy of $X$. Later, Arikan~\cite{Arikan} investigated the problem of bounding the moments of the number of guesses in terms of the R\'{e}nyi entropy~\cite{renyi1961measures} of $X$. Further connections between R\'{e}nyi entropy and guessing are explored in~\cite{SasonV18,SalamatianHBCM19,ChristiansenD13,MerhavC20}.


We study the guessing problem where an adversary makes a fixed number of guesses. Such a setting finds applications in several practical scenarios. For example, an adversary is allowed several guesses to login with a password before getting locked-out. We consider a setup where an adversary is interested in guessing the unknown value of a random variable $X$ on observing another correlated random variable $Y$, where $X$ and $Y$ are jointly distributed according to $P_{XY}$ over the finite support $\mathcal{X}\times\mathcal{Y}$. Since the adversary makes a fixed number of guesses $k$, we focus on evaluating the adversary's success using loss functions that in turn can measure the information leaked by $Y$ about $X$. 
To this end, we model the adversary's strategy using $\alpha$-loss, a class of tunable loss functions parameterized by $\alpha \in (0,\infty]$ ~\cite{LiaoKS20,SypherdDSK19}. This class captures the well-known exponential loss ($\alpha=1/2$)~\cite{FREUND1997119}, log-loss ($\alpha=1$)~\cite{MerhavF1998,NguyenWJ09,CourtadeW11}, and the 0-1 loss ($\alpha=\infty$)~\cite{NguyenWJ09,BartlettJM06}. The adversary then seeks to find the optimal (possibly randomized) guessing strategy that minimizes the expected $\alpha$-loss over $k$ guesses. 

Devising \emph{guessing strategies} with the quest to optimize certain performance metrics of an adversary has several applications in information theory and related fields; this includes sequential decoding~\cite{Arikan}, guessing codewords~\cite{PfisterS04}, botnet attacks~\cite{SalamatianHBCM19,MerhavC20}, to name a few. In~\cite{SalamatianHBCM19}, the authors consider a guessing problem with a fixed number of guesses allowing for randomized guessing strategies (similar to our setting) and analyze the exponential behaviour of the probability of success in guessing the sequences. A closely related work is that of \emph{maximal leakage}~\cite{Issaetal} which captures the information leaked when an adversary maximizes its probability of correctly guessing (equivalent to minimizing $0$-$1$ loss) an unknown function of $X$; they further generalize this notion to $k$-guesses, and they show the resulting leakage measure is unchanged. 



Our main contributions are as follows:
\begin{itemize}
    \item We completely characterize the minimal expected $\alpha$-loss for $k$ guesses (Theorem~\ref{thm:minimalalphaloss}), thereby recovering known results for $\alpha=\infty$~\cite{Issaetal}. 
    To the best of our knowledge, such a result even for log-loss ($\alpha=1$) under multiple guesses was not explored earlier. 
    We derive a technique for transforming the optimization problem over the probability simplex associated with multiple random variables to that of with a single random variable using tools drawn from duality in linear programming, which may be of independent interest (Lemma~\ref{theorem:probclass1}). 
    \item We define a measure of {information leakage} for $k$ guesses of an adversary motivated by $\alpha$-leakage~\cite[Definition~5]{LiaoKS20} and show that it does not change with the number of guesses for a class of probability distributions $P_{XY}$ (Theorem~\ref{thm2}).
\end{itemize}
\section{Background and Problem Definition}\label{section:def}
We first review $\alpha$-loss and then define the minimal expected $\alpha$-loss for $k$ guesses. Later, we define a measure of information leakage based on this.
\begin{definition}[$\alpha$-loss~\cite{LiaoKS20,SypherdDSK19}]
For $\alpha\in(0,1)\cup(1,\infty)$, the $\alpha$-loss is a function defined from $[0,1]$ to $\mathbb{R}_+$ as
\begin{align}
    \ell_\alpha(p):=\frac{\alpha}{\alpha-1}\left(1-p^{\frac{\alpha-1}{\alpha}}\right).
\end{align}
It is defined by continuous extension for $\alpha=1$ and $\alpha=\infty$, respectively, and is given by
\begin{align}
    \ell_1(p)=\log{\frac{1}{p}},\  \ell_\infty(p)=1-p.
\end{align}
\end{definition}
Notice that $\ell_\alpha(p)$ is decreasing in $p$.
\begin{definition}[Minimal expected $\alpha$-loss for $k$ guesses]
Consider random variables $(X,Y)\sim P_{XY}$ and an adversary that makes $k$ guesses $\hat{X}_{[1:k]}=\hat{X}_1,\hat{X}_2,\dots,\hat{X_k}$ 
on observing $Y$ such that $X-Y-\hat{X}_{[1:k]}$ is a Markov chain. Let $P_{\hat{X}_{[1:k]}|Y}$ be a strategy for estimating $X$ from $Y$ in $k$ guesses. For $\alpha\in(0,\infty],$ the minimal expected $\alpha$-loss for $k$ guesses is defined as
\begin{multline}\label{minimallossdef} 
 \mathcal{ME}^{(k)}_\alpha(P_{XY})\\ :=\min_{P_{\hat{X}_{[1:k]}|Y}}
    \sum_{x,y} P_{XY}(x,y) \ell_\alpha\left(\mathrm{P}\left(\bigcup_{i=1}^k (\hat{X}_i=x|Y=y)\right)\right).  
\end{multline}
\end{definition}

We interpret $\mathrm{P}\left(\bigcup\limits_{i=1}^k(\hat{X}_i=x)|Y=y\right)$ as the probability of correctly estimating $X=x$ given $Y=y$ in $k$ guesses. An adversary seeks to find the optimal guessing strategy in \eqref{minimallossdef}.
Note that the optimization problem in \eqref{minimallossdef} was solved for a special case of $k=1$ by Liao \textit{et al.}~\cite[Lemma~1]{LiaoKS20}.
Notice that
\begin{align}\label{eqn:simpli}
    \mathcal{ME}^{(k)}_\alpha(P_{XY})=\sum_yP_Y(y)\mathcal{ME}_\alpha^{(k)}(P_{X|Y=y}),
\end{align}
where we have slightly abused the notation in the R.H.S. of \eqref{eqn:simpli}. Hence, in view of \eqref{eqn:simpli}, in order to solve the optimization problem in \eqref{minimallossdef}, it suffices to solve for a case where $Y=\emptyset$, i.e., 
\begin{equation}\label{eqn:optimizationproblem}
   \mathcal{ME}_\alpha^{(k)}(P_X):= \min_{P_{\hat{X}_{[1:k]}}}
    \sum_x P_X(x) \ell_\alpha\left(P\left(\bigcup_{i=1}^k (\hat{X}_i=x)\right)\right).
\end{equation}
 Also, in the sequel, it suffices to consider the optimization problem in \eqref{eqn:optimizationproblem} only for the case where $k<n$, where $P_X$ is supported on $\mathcal{X}=\{x_1,x_2,\dots,x_n\}$ because if $k\geq n$, we have
 $\mathcal{ME}_\alpha^{(k)}(P_X)=0$,
since a strategy $P^*_{\hat{X}_{[1:k]}}$ such that $P^*_{\hat{X}_{[1:n]}}(x_1,x_2,\dots,x_n)=1$ is optimal.

Motivated by $\alpha$-leakage~\cite[Definition~5]{LiaoKS20} which captures how much information an adversary can learn about a random variable $X$ from a correlated random variable $Y$ when a single guess is allowed, we define a leakage measure which captures the information an adversary can learn when $k$ guesses are allowed. This definition is also related to maximal leakage under $k$ guesses \cite{Issaetal}.
\begin{definition}[$\alpha$-leakage with $k$ guesses]
Given a joint distribution $P_{XY}$ and $k$ estimators $\hat{X}_1,\hat{X}_2,\dots,\hat{X}_k$ with the same support as $X$, the $\alpha$-leakage from $X$ to $Y$ with $k$ guesses is defined as 
\begin{multline}\label{defn:kalphaleakage}
    \mathcal{L}^{(k)}_\alpha(X\rightarrow Y)\triangleq\\ \frac{\alpha}{\alpha-1}\log{\frac{\max\limits_{P_{\hat{X}_{[1:k]}|Y}}\mathbb{E}\left[\mathrm{P}\left(\bigcup\limits_{i=1}^k(\hat{X}_i=X)|Y\right)^{\frac{\alpha-1}{\alpha}}\right]}{\max\limits_{P_{\hat{X}_{[1:k]}}}\mathbb{E}\left[\mathrm{P}\left(\bigcup\limits_{i=1}^k(\hat{X}_i=X)\right)^{\frac{\alpha-1}{\alpha}}\right]}},
\end{multline}
for $\alpha\in(0,1)\cup (1,\infty)$.
\end{definition}
\section{Main Results}\label{section:mainresults}
\begin{theorem}[Minimal expected $\alpha$-loss for $k$ guesses]\label{thm:minimalalphaloss}
Consider a $P_X$ supported on $\mathcal{X}=\{x_1,x_2,\dots,x_n\}$ such that $p_1\geq p_2\geq \dots \geq p_n$, where $p_i:=P_X(x_i)$, for $i\in[1:n]$. Then the minimal expected $\alpha$-loss for $k$ guesses is given by
\begin{align}\label{eqn:thm11}
    \mathcal{ME}^{(k)}_\alpha(P_X)= \frac{\alpha}{\alpha-1}\sum\limits_{i=s^*}^np_i\left(1-\left(\frac{(k-s^*+1)p_i^\alpha}{\sum_{j=s^*}^np_j^\alpha}\right)^\frac{\alpha-1}{\alpha}\right),
\end{align}
where 
\begin{align}\label{eqn:thm1sstar}
    s^*=\min\left\{r\in\{1,2,\ldots,k\}: \frac{(k-r+1)p_r^\alpha}{\sum_{i=r}^np_i^\alpha}\leq 1\right\}.
\end{align}
\end{theorem}
\begin{remark}
It can be inferred from Theorem~\ref{thm:minimalalphaloss} that in the optimal guessing strategy, the adversary guesses the $s^*-1$ most likely outcomes, and uses an updated tilted distribution on the rest of the outcomes (see also \eqref{eqn:newcam}). For the special case when $k=s^*=2$, this optimal strategy is exactly the same as that of a seemingly different guessing problem considered in \cite[Section~II-B]{HuleihelSM17}.
\end{remark}
\begin{remark}
Notice that whenever $s^*=1$ in \eqref{eqn:thm1sstar}, the expression in \eqref{eqn:thm11} simplifies to
\begin{align}\label{eqn:thm1recover}
     \frac{\alpha}{\alpha-1}\left(1-k^{\frac{\alpha-1}{\alpha}}\exp{\left(\frac{1-\alpha}{\alpha}H_\alpha(X)\right)}\right),
\end{align}
where $H_{\alpha}(X)=\frac{1}{1-\alpha}\log{\left(\sum_{i=1}^np_i^\alpha\right)}$ is the R\'{e}nyi entropy of order $\alpha$~\cite{renyi1961measures}. Also, note that for the special case of $k=1$, we always have $s^*=1$, thereby recovering \cite[Lemma~1]{LiaoKS20}.
\end{remark}
\begin{corollary}[Minimal expected log-loss \{$\alpha=1$\} for $k$ guesses]\label{corollary1}
Under the notations of Theorem~\ref{thm:minimalalphaloss}, the minimal expected log-loss for $k$ guesses is given by
\begin{align}
    \mathcal{ME}^{(k)}_1(P_X)&=H(X)-H_{s^*}\left(p_1,p_2,\dots,p_{{s^*}-1},\sum_{i={s^*}}^np_i\right)\nonumber\\
    &\hspace{24pt}-\left(\sum\limits_{i={s^*}}^np_i\right)\log{(k-{s^*}+1)},
\end{align}
where $s^*=\min\left\{r\in\{1,2,\ldots,k\}: \frac{(k-r+1)p_r}{\sum_{i=r}^np_i}\leq 1\right\}$
and $H_{s^*}(q_1,q_2,\dots,q_{s^*}):=\sum_{i=1}^{s^*}q_i\log{\frac{1}{q_i}}$ is the entropy function.
\end{corollary}
\begin{corollary}[Minimal expected $0$-$1$ loss \{$\alpha=\infty$\} for $k$ guesses]\label{corollary2}
Under the notations of Theorem~\ref{thm:minimalalphaloss}, the minimal expected $0$-$1$ loss for $k$ guesses is given by
\begin{align}
   \mathcal{ME}^{(k)}_\infty(P_X)
    &=1-\sum_{i=1}^kp_i\nonumber\\
    &=1-\max_{\substack{a_1,a_2,\dots,a_k:\\ a_l\neq a_m,l\neq m}}\sum_{i=1}^kP_X(a_i).
\end{align}
\end{corollary}
The following theorem shows the robustness of $\alpha$-leakage
to the number of guesses for a class of probability distributions $P_{XY}$. Let $P_{X|Y=y}^{(\alpha)}$ denote the tilted distribution of $P_{X|Y=y}$, i.e., $P_{X|Y}^{(\alpha)}(x|y)=\frac{P_{X|Y}(x|y)^\alpha}{\sum_xP_{X|Y}(x|y)^\alpha}$.
\begin{theorem}[Robustness of $\alpha$-leakage to number of guesses]\label{thm2}
Consider a $P_{XY}$ such that $P_{X|Y}^{(\alpha)}(x|y)\leq \frac{1}{k}$, for all $x,y$ and $P_{X}^{(\alpha)}(x)\leq \frac{1}{k}$, for all $x$. Then
\begin{align}
    \mathcal{L}^{(k)}_\alpha=\mathcal{L}^{(1)}_\alpha.
\end{align}
\end{theorem}

The proofs of Theorems~\ref{thm:minimalalphaloss} and~\ref{thm2} are given in the following section.

\section{Proofs of Main Results}\label{proofs}
We begin with the following lemmas which will be useful in the proof of Theorem~\ref{thm:minimalalphaloss}. It is intuitive to expect that an optimal strategy, $P^*_{\hat{X}_{[1:k]}}$, puts zero weight on ordered tuples $(a_1,a_2,\dots,a_k)$ (denoted as $a_{[1:k]}$ in the sequel) whenever $a_i=a_j$ for some $i\neq j$, since there is no advantage in guessing the same estimate more than once. The following lemma based on the monotonicity of the $\alpha$-loss formalizes this.
\begin{lemma}\label{fact:alphaloss2noneq}
If $P^*_{\hat{X}_{[1:k]}}$ is an optimal strategy for the optimization problem in \eqref{eqn:optimizationproblem}, then 
\begin{align*}
    P^*_{\hat{X}_{[1:k]}}(a_{[1:k]})=0, \ \text{for all} \ a_{[1:k]} \ \text{s.t.} \ a_i=a_j, \ \text{for some} \ i\neq j.
\end{align*}
\end{lemma}
The proof of Lemma~\ref{fact:alphaloss2noneq} is deferred to Appendix~\ref{app1}.
\begin{remark}\label{remark}
An important consequence of Lemma~\ref{fact:alphaloss2noneq} is that, if $P^*_{\hat{X}_{[1:k]}}$ is an optimal strategy for the optimization problem in \eqref{eqn:optimizationproblem}, then we have
\begin{align}\label{eqn:nonequality}
    \sum_x\mathrm{P}^*\left(\bigcup_{i=1}^k(\hat{X}_i=x)\right)=k,
\end{align}
where the probability $\mathrm{P}^*$ is taken with respect to an optimal strategy $P^*_{\hat{X}_{[1:k]}}$. Hence, it suffices to consider the optimization in \eqref{eqn:optimizationproblem} over all the strategies $P_{\hat{X}_{[1:k]}}$ satisfying \eqref{eqn:nonequality}.
\end{remark}
Let $\mathcal{X}=\{x_1,x_2,\dots,x_n\}$ be the support of $P_X$. A vector $(t_1,t_2,\dots,t_n)$ such that $\sum_{i=1}^nt_i=k$ is said to be \emph{admissible} if there exists a strategy $P_{\hat{X}_{[1:k]}}$ satisfying
\begin{align}\label{eqn:systemprob}
    t_i=\mathrm{P}\left(\bigcup_{j=1}^k(\hat{X}_j=x_i)\right ),\ \text{for all}\ i\in[1:n].
\end{align}
Equivalently, \eqref{eqn:systemprob} can be written as the following system of linear equations.
\begin{align}\label{eqn:linsystemk}
    t_i=\sum_{a_{[1:k]}:\bigcup\limits_{j=1}^k(a_j=x_i)} P_{\hat{X}_{[1:k]}}(a_{[1:k]}), \ \text{for all}\ i\in[1:n].
\end{align}
 In general, in order to determine whether a vector $(t_1,t_2,\dots,t_n)$ is admissible or not, we need to solve a linear programming problem (LPP) with number of variables and constraints that are polynomial in the support size of $P_X$, i.e, $n$. Nonetheless, the following lemma based on Farkas' lemma~\cite[Proposition 6.4.3]{Matousek} completely characterizes the necessary and sufficient conditions for the admissibility of a vector $(t_1,t_2,\dots,t_n)$.
\begin{lemma}\label{theorem:probclass1}
A vector $(t_1,t_2,\dots,t_n)$ such that $\sum\limits_{i=1}^nt_i=k$ is admissible if and only if $0\leq t_i\leq 1$, for all $i\in[1:n]$.
\end{lemma}
The proof of Lemma~\ref{theorem:probclass1} is deferred to Appendix~\ref{app2}. We are now ready to prove Theorem~\ref{thm:minimalalphaloss}.
\begin{proof}[Proof of Theorem~\ref{thm:minimalalphaloss}]
From the definition of the minimal expected $\alpha$-loss for $k$ guesses in \eqref{eqn:optimizationproblem}, we have 
\begin{align}
    &\mathcal{ME}^{(k)}_\alpha(P_X)\nonumber\\
    &=\min_{P_{\hat{X}_{[1:k]}}}\frac{\alpha}{\alpha-1}\text{\small$\left[\sum_{i=1}^np_i\left(1-\mathrm{P}\left(\bigcup_{j=1}^k(\hat{X}_j=x_i)\right)^\frac{\alpha-1}{\alpha}\right)\right]$}\\
    &=\min_{P_{\hat{X}_{[1:k]}}}\frac{\alpha}{\alpha-1}\text{\small$\left[\sum_{i=1}^np_i\left(1-\mathrm{P}\left(\bigcup_{j=1}^k(\hat{X}_j=x_i)\right)^\frac{\alpha-1}{\alpha}\right)\right]$}\nonumber\\
  &\hspace{24pt}\text{s.t.}\ \sum\limits_{i=1}^n\mathrm{P}\left(\bigcup_{j=1}^k(\hat{X}_j=x_i)\right)=k\label{eqn:thmnoneq}\\
  &=\min_{t_1,\dots,t_n}\frac{\alpha}{\alpha-1}\left[\sum_{i=1}^np_i(1-t_i^{\frac{\alpha-1}{\alpha}})\right]\nonumber\\
  &\hspace{24pt}\text{s.t.}\ \sum_{i=1}^nt_i=k,\nonumber\\
  &\hspace{43pt}0\leq t_i\leq 1,\ i\in[1:n]\label{eqn:thmduality1},
\end{align}
where \eqref{eqn:thmnoneq} follows from Lemma~\ref{fact:alphaloss2noneq} and Remark~\ref{remark}, and \eqref{eqn:thmduality1} follows from  the change of variable $t_i=\mathrm{P}\left(\bigcup\limits_{j=1}^k(\hat{X}_j=x_i)\right)$ and Lemma~\ref{theorem:probclass1}.
Consider the Lagrangian 
\begin{multline}
    \mathcal{L}=\frac{\alpha}{\alpha-1}\left[\sum_{i=1}^np_i(1-t_i^{\frac{\alpha-1}{\alpha}})\right]+\lambda\left(\sum_{i=1}^nt_i-k\right)\\
    +\sum_{i=1}^n\mu_i(t_i-1)
\end{multline}
The Karush-Kuhn-Tucker (KKT) conditions~\cite[Chapter~5.5.3]{boyd_vandenberghe_2004} are given by
\begin{align}
    &\text{(Stationarity):}\ \frac{\partial\mathcal{L}}{\partial t_i}=0,i\in[1:n],\nonumber\\ &\text{i.e.,}\ t_i=\left(\frac{p_i}{\lambda+\mu_i}\right)^\alpha,i\in[1:n],\label{stationarity}\\
   &\text{(Primal feasibility):}\ \sum_{i=1}^nt_i=k,
    0\leq t_i\leq 1, i\in[1:n],\label{eqn:primalfeasibility}\\
   &\text{(Dual feasibility):}\ \mu_i\geq 0,i\in[1:n],\label{dualfeasibility}\\
&\text{(Complementary slackness):}\ \mu_i(t_i-1)=0, i\in[1:n]\label{slackness}.
\end{align}
Notice that for $\alpha>1$, $t^{\frac{\alpha-1}{\alpha}}$ is a concave function of $t$, meaning the overall objective function in \eqref{eqn:thmduality1} is convex. For $\alpha<1$, $t^{\frac{\alpha-1}{\alpha}}$ is a convex function of $t$, but since $\frac{\alpha}{\alpha-1}$ is negative, the overall function is again convex. Thus \eqref{eqn:thmduality1} amounts to a convex optimization problem. Now since KKT conditions are necessary and sufficient conditions for optimality in a convex optimization problem, it suffices to find values of $t_i$, $i\in[1:n]$, $\lambda$, $\mu_i$, $i\in[1:n]$ satisfying \eqref{stationarity}--\eqref{slackness} in order to solve the optimization problem~\eqref{eqn:thmduality1}. 

First we simplify the KKT conditions \eqref{stationarity}--\eqref{slackness} in the following manner. 
\begin{itemize}
\item For $i$ such that $\left(\frac{p_i}{\lambda}\right)^\alpha\leq 1$, we take $\mu_i=0$ and $t_i=\left(\frac{p_i}{\lambda}\right)^\alpha$.
\item For $i$ such that $\left(\frac{p_i}{\lambda}\right)^\alpha> 1$, we take $\mu_i=p_i-\lambda$ and $t_i=1$. Notice that for such $i$, we have $\mu_i>0$, since $p_i>\lambda$.
\end{itemize}
This is equivalent to choosing $t_i=\min\left\{\left(\frac{p_i}{\lambda}\right)^\alpha,1\right\}$ and $\mu_i=0$ or $\mu_i=p_i-\lambda$ depending on whether $t_i=\left(\frac{p_i}{\lambda}\right)^\alpha$ or $t_i=1$, respectively, for each $i\in[1:n]$. Notice that this choice is consistent with the KKT conditions \eqref{stationarity}--\eqref{slackness} except for that $\lambda$ has to be chosen appropriately satisfying $\sum_{i=1}^nt_i=k$ also. In effect, we have essentially reduced the KKT conditions \eqref{stationarity}--\eqref{slackness} to the following equations by eliminating $\mu_i$'s:
\begin{align}
    &t_i=\min\left\{\left(\frac{p_i}{\lambda}\right)^\alpha,1\right\},i\in[1:n],\label{eqn:new1}\\
    &\sum_{i=1}^nt_i=k\label{eqn:new2}.
\end{align}

We solve the equations \eqref{eqn:new1} and \eqref{eqn:new2} by considering the following $k$ mutually exclusive and exhaustive cases (clarified later) based on $P_X$. 

\medskip
\noindent \underline{Case $1$}  $\left(\frac{p_1^\alpha}{\sum_{i=1}^np_i^\alpha}\leq \frac{1}{k}\right)$:\\
Consider the choice
\begin{align}
\lambda=\left(\frac{\sum_{i=1}^np_i^\alpha}{k}\right)^\frac{1}{\alpha},\ 
    t_i=\frac{kp_i^\alpha}{\sum_{j=1}^np_j^\alpha},i\in[1:n].
\end{align}
This choice satisfies \eqref{eqn:new1} and \eqref{eqn:new2} since $\frac{kp_1^\alpha}{\sum_{i=1}^np_i^\alpha}\leq 1$ and $p_1\geq p_2\dots\geq p_n$. 

\bigskip

\noindent\underline{Case `$s$'} ($2\leq s\leq k$) $\left(\frac{(k-s+2)p_{s-1}^\alpha}{\sum_{i=s-1}^np_i^\alpha}>1, \frac{(k-s+1)p_s^\alpha}{\sum_{i=s}^np_i^\alpha}\leq 1\right)$:\\
Consider the choice
\begin{align}
    \lambda&=\left(\frac{\sum_{i=s}^np_i^\alpha}{k-s+1}\right)^\frac{1}{\alpha},\\
    t_i&=1, i\in[1:s-1], t_i=\frac{(k-s+1)p_i^\alpha}{\sum_{j=s}^np_j^\alpha}, i\in[s:n]\label{eqn:newcam}.
\end{align}
This choice satisfies \eqref{eqn:new1}
\begin{itemize}
    \item for $i\in[1:s-1]$ because $\frac{(k-s+2)p_{s-1}^\alpha}{\sum_{i=s-1}^np_i^\alpha}>1$ and $p_1\geq p_2\geq\dots\geq p_{s-1}$, and
    \item for $i\in[s:n]$ because $\frac{(k-s+1)p_s^\alpha}{\sum_{i=s}^np_i^\alpha}\leq 1$ and $p_s\geq p_{s+1}\geq\dots\geq p_n$.
\end{itemize} 
Also, this choice clearly satisfies \eqref{eqn:new2}. 
Finally, notice that the condition for Case `$s$', $2\leq s\leq n$, can be written as 
\begin{align}
    \frac{(k-i+1)p_{i}^\alpha}{\sum\limits_{j=i}^np_j^\alpha}>1, \ \text{for}\ i\in[1:s-1], \frac{(k-s+1)p_s^\alpha}{\sum\limits_{i=s}^np_i^\alpha}\leq 1
\end{align}
since $\frac{(k-s+2)p_{s-1}^\alpha}{\sum\limits_{i=s-1}^np_i^\alpha}>1$ and $p_1\geq p_2\geq\dots\geq p_{s-1}$. This proves that the cases considered above are mutually exclusive and exhaustive, and together with the case-wise analysis gives the expression for the minimal expected $\alpha$-loss for $k$ guesses as presented in Theorem~\ref{thm:minimalalphaloss}.
\end{proof}
The proof of Corollary~\ref{corollary1} follows by taking limit $\alpha\rightarrow 1$ using L'H\^{o}pital's rule in the result of Theorem~\ref{thm:minimalalphaloss} and rearranging the terms. The proof of Corollary~\ref{corollary2} follows by taking limit $\alpha\rightarrow \infty$ in Theorem~\ref{thm:minimalalphaloss}.
\begin{proof}[Proof of Theorem~\ref{thm2}]
From the definition of $\alpha$-leakage with $k$ guesses in \eqref{defn:kalphaleakage}, we have
\begin{align}
    &\mathcal{L}^{(k)}_\alpha(X\rightarrow Y)\nonumber\\ &=\frac{\alpha}{\alpha-1}\log{\frac{\max\limits_{P_{\hat{X}_{[1:k]}|Y}}\mathbb{E}\left[\mathrm{P}\left(\bigcup\limits_{i=1}^k(\hat{X}_i=X)|Y\right)^{\frac{\alpha-1}{\alpha}}\right]}{\max\limits_{P_{\hat{X}_{[1:k]}}}\mathbb{E}\left[\mathrm{P}\left(\bigcup\limits_{i=1}^k(\hat{X}_i=X)\right)^{\frac{\alpha-1}{\alpha}}\right]}}\\
    &=\frac{\alpha}{\alpha-1}\log{\frac{k^{\frac{\alpha-1}{\alpha}}\exp{(\frac{1-\alpha}{\alpha}H_\alpha^A(X|Y))}}{k^{\frac{\alpha-1}{\alpha}}\exp{(\frac{1-\alpha}{\alpha}H_\alpha(X))}}}\label{eqn:leakage1}\\
     &=\frac{\alpha}{\alpha-1}\log{\frac{\exp{(\frac{1-\alpha}{\alpha}H_\alpha^A(X|Y))}}{\exp{(\frac{1-\alpha}{\alpha}H_\alpha(X))}}}\\
    &=\mathcal{L}_\alpha^{(1)},
\end{align}
where \eqref{eqn:leakage1} follows from Theorem~\ref{thm:minimalalphaloss}, in particular from the case when $s^*=1$ since $P_{X|Y}^{(\alpha)}(x|y)\leq \frac{1}{k}$, for all $x,y$ and $P_{X}^{(\alpha)}(x)\leq \frac{1}{k}$, for all $x$, and $H_\alpha^A(X|Y)$ in \eqref{eqn:leakage1} is the Arimoto conditional entropy~\cite{arimoto1977information} defined as $H_\alpha^A(X|Y)=\frac{\alpha}{1-\alpha}\log{\sum\limits_y\left(\sum\limits_xP_{XY}(x,y)^\alpha\right)^\frac{1}{\alpha}}$.
\end{proof}
\section{Conclusion}\label{section:conclusion}
There are many questions to be further studied. For example, analogously to maximal leakage \cite{Issaetal} and maximal $\alpha$-leakage \cite{LiaoKS20}, we can define a maximal version of $\alpha$-leakage with $k$ guesses. As shown in \cite{Issaetal}, for $\alpha=\infty$, this quantity does not change with $k$; it would be interesting to understand whether this is also true for other $\alpha$.

\appendices
\section{Proof of Lemma~\ref{fact:alphaloss2noneq}}\label{app1}
\balance
Let $\mathcal{X}=\{x_1,x_2,\dots,x_n\}$ and $P_X(x_i)=p_i$, for $i\in[1:n]$. Consider $a_{[1:k]}$ such that $a_i=a_j$ for some $i\neq j$. There exists a $b_{[1:k]}$ such that for each $i\in[1:k]$, we have $a_i=b_j$ for some $j$ and $b_r\neq a_j$ for some $r$ and any $j$. Consider
\begin{align}
    \frac{\alpha}{\alpha-1}\text{\small$\left[\sum_{i=1}^np_i\left(1-\mathrm{P}^*\left(\bigcup_{j=1}^k(\hat{X}_j=x_i)\right)^\frac{\alpha-1}{\alpha}\right)\right]$}.\label{eqn:appendix}
\end{align}
Let $\mathcal{A}$ and $\mathcal{B}$ denote the sets of all multiset permutations of $a_{[1:k]}$ and $b_{[1:k]}$, respectively, when $a_{[1:k]}$ and $b_{[1:k]}$ are treated as  multisets. Let $q_{a_1,a_2,\dots,a_k}:=\sum_{r_{[1:k]}\in\mathcal{A}}P_{\hat{X}_{[1:k]}}(r_{[1:k]})$ and $q_{b_1,b_2,\dots,b_k}:=\sum_{r_{[1:k]}\in\mathcal{B}}P_{\hat{X}_{[1:k]}}(r_{[1:k]})$.
Each term out of the $n$ terms in \eqref{eqn:appendix} will either contain both $q_{a_{[1:k]}}$ and $q_{b_{[1:k]}}$ (say, type 1), contain just $q_{b_{[1:k]}}$ alone (say, type 2), or does not contain both (say, type 3). We now construct a new strategy $P_{\hat{X}_{[1:k]}}$ by incorporating the value of $q_{a_{[1:k]}}$ into $q_{b_{[1:k]}}$ making the value of new $q_{a_{[1:k]}}$ equal to zero. Now the values of the terms of type 2 strictly decrease as the $\alpha$-loss function is strictly decreasing in its argument while retaining the values of the terms of types 1 and 3. This leads to a contradiction since $P^*_{X_{[1:k]}}$ is assumed to be an optimal strategy. So, $P_{\hat{X}_{[1:k]}}(a_{[1:k]})=0$. Repeating the same argument as above for all such $a_{[1:k]}$ s.t. $a_i=a_j$, for some $i\neq j$ completes the proof. 
\section{Proof of Lemma~\ref{theorem:probclass1}}\label{app2}
 \noindent\underline{`Only if' part}: Suppose a vector $(t_1,t_2,\dots,t_n)$ is admissible. Then there exists $P_{\hat{X}_{[1:k]}}$ satisfying \eqref{eqn:linsystemk}. Using \eqref{eqn:systemprob}, since $t_i$ is probability of a certain event, we have 
 \begin{align*}
     0\leq t_i \leq 1,\ \text{for}\ i\in[1:n]. 
 \end{align*}
 \underline{`If' part}: Suppose $0\leq t_i\leq 1$, for $i\in[1:n]$. Summing up all the equations in \eqref{eqn:linsystemk} over $i\in[1:n]$ and using $\sum_{i=1}^nt_i=k$, we get
 \begin{align*}
     P_{\hat{X}_{[1:k]}}(a_{[1:k]})=0,\ \text{for all}\ a_{[1:k]}\ \text{s.t.}\ a_i=a_j, \ \text{for some}\ i\neq j.
 \end{align*}
 With this, \eqref{eqn:linsystemk} can be written in the form of system of linear equation only in terms of non-negative variables of the form 
 \begin{align}
 q_{i_1,i_2,\dots,i_k}:=\sum\limits_{\sigma\in S_n}P_{\hat{X}_{[1:k]}}(x_{i_{\sigma(1)}},x_{i_{\sigma(2)}},\dots,x_{i_{\sigma(n)}}),
 \end{align}
 where $i_1,i_2,\dots,i_k$ are all distinct and belong to $[1:n]$. Here the sum is computed over all the permutations $\sigma$ of the set $\{1,2,\dots,n\}$. The set of all such permutations is denoted by $S_n$. With this, the system of equations in \eqref{eqn:linsystemk} can be written in the form $AQ=b$, $Q\geq 0$. Here $A$ is a $n\times\binom{n}{k}$-matrix, where the rows are indexed by $i\in[1:n]$ and columns are indexed by $(i_1,i_2,\dots,i_k)$, where $i_1,i_2,\dots,i_k$ are all distinct and belong to $[1:n]$. In particular, in the column indexed by $(i_1,i_2,\dots,i_k)$, the entry of $A$ corresponding to $i_j^{\text{th}}$ row is $1$, for $j\in[1:k]$. All the remaining entries of the matrix $A$ are zeros. $Q$ is $\binom{n}{k}$-length vector of variables of the form $q_{i_1,i_2,\dots,i_k}$. ${b}$ is an $n$-length vector with $b_i=t_i$. We are interested in the feasibility of the system $AQ=b$, $Q\geq 0$. We use the Farkas' lemma~\cite[Proposition 6.4.3]{Matousek} in linear programming for checking this. It states that the system $A{Q}={b}$ has a non-negative solution if and only if every ${y}\in\mathbbm{R}^n$ with ${y}^\top A\geq 0$ also implies ${y}^\top {b}\geq 0$. For our problem, ${y}^\top A\geq 0$ is equivalent to 
 \begin{align}\label{eqn:farkashyp}
     \sum_{j=1}^ky_{i_j}\geq 0,\ \text{for all distinct}\ i_1,i_2,\dots,i_k\in[1:n]. 
 \end{align}
 Without loss of generality, let us assume that $y_i\leq y_{i+1}$, $i\in[1:n-1]$. Then \eqref{eqn:farkashyp} is equivalent to
 \begin{align}\label{eqn:farkashyp1}
     \sum_{i=1}^ky_i\geq 0.
 \end{align}
 Now consider
 \begin{align}
     &\sum_{i=1}^ny_it_i\nonumber\\
     &=\sum_{i=1}^ky_it_i+y_{k+1}t_{k+1}+\sum_{i=k+2}^ny_it_i\\
     &=\sum_{i=1}^k{y_i}+\sum_{i=1}^ky_i(t_i-1)+y_{k+1}t_{k+1}+\sum_{i=k+2}^ny_it_i\\
     &\geq\sum_{i=1}^ky_i+y_{k+1}\sum_{i=1}^k(t_i-1)+y_{k+1}t_{k+1}+\sum_{i=k+2}^ny_it_i\label{theorem:probclass1proof1}\\
     &\geq\sum_{i=1}^k{y_i}+y_{k+1}\sum_{i=1}^k(t_i-1)+y_{k+1}t_{k+1}+y_{k+1}\sum_{i=k+2}^nt_i\label{theorem:probclass1proof2}\\
     &=\sum_{i=1}^ky_i+y_{k+1}\left(\sum_{i=1}^nt_i-k\right)\\
     &={\sum_{i=1}^ky_i}\label{theorem:probclass1proof3}\\
     &\geq 0\label{theorem:probclass1proof4},
 \end{align}
 where \eqref{theorem:probclass1proof1} follows because $y_i\leq y_{k+1}$ and $t_i-1\leq 0$, for $i\in[1:k]$, \eqref{theorem:probclass1proof2} follows because $y_i\geq y_{k+1}$, for $i\in[k+2:n]$, and \eqref{theorem:probclass1proof3} follows because $\sum_{i=1}^nt_i=k$, \eqref{theorem:probclass1proof4} follows from \eqref{eqn:farkashyp1}. Now using the Farkas' lemma, $AQ=b$, has a non-negative solution, i.e., the vector $(t_1,t_2,\dots,t_n)$ is admissible.
\bibliographystyle{IEEEtran}
\bibliography{Bibliography}
\end{document}